\documentclass[conference]{IEEEtran}
\usepackage[]{algorithm2e}
\usepackage{amsmath,amssymb,cite,amsthm,graphicx,array,epsfig,epstopdf}

\usepackage[mathscr]{euscript}
\theoremstyle{definition}
\newtheorem{definition}{Definition}
\newtheorem{example}{Example}

\theoremstyle{plain}
\newtheorem{theorem}{Theorem}

\newtheorem{lemma}{Lemma}
\newtheorem{corollary}{Corollary}
\newtheorem{proposition}{Proposition}

\theoremstyle{remark}

\usepackage{color}

\usepackage{caption}
\usepackage{csquotes}
\usepackage{balance}
\usepackage{arydshln}

\usepackage[mathscr]{euscript}
\title{ Unicast-Uniprior Index Coding Problems:\\ Minrank and Criticality}
\begin{document}

\title{ Unicast-Uniprior Index Coding Problems:\\ Minrank and Criticality}
 \author{
   \IEEEauthorblockN{Niranjana Ambadi}
   \IEEEauthorblockA{Department of Electrical Communication Engineering,\\
    Indian Institute of Science, Bangalore, India  560012.\\Email: ambadi@iisc.ac.in.}
 }
\maketitle

\thispagestyle{empty}
\begin{abstract}
	
An index coding problem is called unicast-uniprior when each receiver demands  a unique subset of messages while knowing another unique subset of messages apriori as side-information. In this work, we give an algorithm to compute the minrank of a unicast-uniprior problem. The proposed algorithm greatly simplifies the computation of minrank for unicast-uniprior  problem settings, over the existing method whose complexity is exponential in the number of messages. First, we establish some  properties that are exclusive to the fitting matrix of a unicast-uniprior  problem. Identification of the critical side information bits follows as a result of this exercise. Further, these properties are used  to lay down the algorithm that computes the  minrank. 
\end{abstract}

	\section{Introduction}
	
	\subsection{Background}
	
	The problem of index coding with side information was introduced by Birk et al. \cite{Birk1,ICSI}. It was motivated by applications in which a server broadcasts a set of messages to a set of clients such as audio and video-on-demand, daily newspaper delivery etc. In practical scenarios, the clients fail to receive all the transmissions broadcast by the server. This can be due to various reasons like erraneous forward channel, limited storage at the clients or path loss effects in the case of a wireless broadcast channel. Through a slow backward channel, the clients communicate to the server the current messages in their caches. The server then exploits this information and encodes the pending demands with the objective of minimizing the number of further transmissions. 
	

	Index coding  is a special case of the well known network coding problem\cite{Ashwelde, survey}. Rouayheb et al. \cite{Rsprint} established a reduction that maps any instance of the network coding problem to a corresponding instance of the index coding problem. 
	
	Very often, an index coding problem where the source has $n$ messages and each message is desired by exactly one receiver is specified by a simple directed graph on $n$ nodes, called the $\textit{side-information graph}$, $G(V,E)$. Each node $v_i \in {V}$ corresponds to a demanded message $x_i$ and there exists an edge $e_{ij} \in E$ from $v_i \rightarrow v_j$ if and only if the receiver $D_j$ knows the message $x_i$ as side-information (see Example \ref{exmp1}). Further, a graph functional called \emph{minrank} was shown to characterize the minimum length of index codes \cite{ICSI}. To this effect, the notion of a binary matrix called $\textit{fitting matrix}$ was introduced. A matrix $A=[a_{ij}]_{n \times n}$ fits $G$ if: 
	\begin{equation}
	a_{ij} =
	\begin{cases}
	1, & \text{if $i=j$; },
	\\
	0, & \text{there is no edge from $v_i$ to $v_j$ in ${G}$},
	\\
	x, & \text{ $x\in \{0,1\}$, otherwise} .
	\end{cases}
	\label{fiteqn}
	\end{equation}
	Let $\kappa(i)$ denote the out-degree of the vertex $v_i$ in $G$. Any of the  $z=2^{\sum\limits_{i=1}^{n} \kappa(i)}$ matrices defined according to $(\ref{fiteqn})$ can be regarded as  a fitting matrix (see Example \ref{exmp1}). Minrank is defined as the minimum among the ranks of all such fitting matrices.
	\begin{equation}
	\text{minrank} \triangleq \text{min} \{rk_2(A); A \; \text{fits} \;  G\},\nonumber
	\end{equation}
	where $rk_2(A)$ denotes the rank over Galois Field of two elements, $GF(2)$. 
	

	Finding the minrank requires  computation of ranks of each of the $z$ fitting  matrices over $GF(2)$. Using the traditional method of  Gaussian elimination to compute ranks, the complexity \cite{BunHop, Ibarra} of computing the minrank becomes   $zO(n^2)$. This was shown to be NP-hard \cite{Bar, Peeters}.
	
However, for some very specific problem settings minrank can be computed in polynomial time \cite{critical,KNR}. In \cite{critical},  the notion of critical graphs in index coding was introduced and it was stated that certain side information bits can be called \emph{critical}, if  removing the corresponding edges in the graph $G$ strictly reduces the \emph{capacity region} (see Section \ref{back1}) of the index coding problem. Critical graphs are defined as minimal graphs that support a given set of rates for a given  index coding problem\cite{critical}. This minimality is defined in terms of the containment of the edge set. 

In many broadcast scenarios,  the network topology plays a crucial role in determining the side-information at receivers \cite{TI}. One might often be able to predict the side-information pattern at each receiver based on the network topology. Whenever the receivers have limited storage, it becomes important to store only those side-information bits which are critical so that index coding can be enabled that do not compromise on  optimal  receiver rates. Thus, identifying critical side-information is useful in computing the minrank of an index coding problem. 
	
\subsection{System Model}
	
	Ong and Ho\cite{OngHo} did a categorization of index coding problems (refer Fig. \ref{fig1}) and we use it as the reference for problem formulation. 
	
	There is a server  with $n$ messages, $\mathcal{M}=\{x_1, x_2, \ldots, x_n\}$ that are to be transmitted to  $N$ receivers, $D_1, D_2, \ldots, D_N$. Each receiver $D_i$ is identified  by the ordered pair $(\mathcal{W}_i, \mathcal{A}_i)$, where $\mathcal{W}_i \subseteq \mathcal{M}$ is the set of messages demanded by $D_i$ and $\mathcal{A}_i \subset \mathcal{M}$ is the set of messages known apriori to $D_i$ as \emph{side-information}.  	
	
The index coding problem is said to be \emph{unicast} if $\mathcal{W}_i \cap \mathcal{W}_j= \phi;\; \forall i\neq j,\; i,j \in [N]$ \footnote{For any positive integer $K$, the set $\{1,2, \ldots, K\}$  is denoted by $[ K ]$.
}. The problem is called \emph{uniprior} when $\mathcal{A}_i \cap \mathcal{A}_j= \phi;\; \forall i\neq j,\; i,j \in [N]$. Any problem that is both unicast and uniprior is called a \emph{unicast-uniprior} problem. 
	
	
			 \begin{figure}[htbp]
	 	\centering
	 	\includegraphics[scale=0.55]{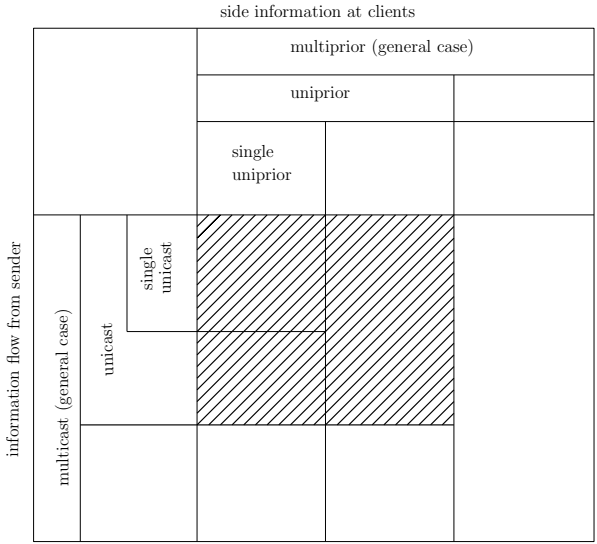}
	 	\caption{Categorization of index coding problems; the shaded part denotes the class of unicast-uniprior  problems that are dealt with in this paper.  }
	 	\label{fig1}
	 \end{figure}

\subsection{Contributions}		
	

	In this work,  the unicast-uniprior  problem is studied.  
	\begin{itemize}
\item The notion of a \emph{critical fitting matrix} for an index coding problem, is introduced. 
\item A new graph structure called the \emph{side-information supergraph} is introduced in order to study the special properties associated with the side-information graph of a unicast-uniprior problem.

\item  The properties of fitting matrix of a unicast-uniprior problem are stated and proved through a series of lemmas.
\item The critical side information for a unicast-uniprior problem is identified using critical graphs and the corresponding fitting matrices.  
\item A new algorithm to compute the minrank is given. This algorithm greatly reduces the number of computations over the existing brute force technique. 
\end{itemize}
 \subsection{Organization}
 The remainder of the paper is organized as follows: The necessary background and terminologies are explained briefly in Section \ref{sec1}. The notions of a critical fitting matrix and a side-information supergraph are introduced in Sections \ref{criticalfit} and  \ref{sidesuper} respectively. Section \ref{sub1} establishes certain properties of a fitting matrix of a unicast-uniprior problem. Section \ref{sub2} contains results concerning the critical graphs and side-information of a unicast-uniprior problem. Finally, the polynomial time algorithm to compute minrank is presented in Section \ref{algo}. Some illustrative examples for minrank computation are given in Section \ref{examples}. Concluding remarks constitute Section \ref{concl}.

\section{Background}
\label{sec1}
This section provides  an overview of some existing results and terminologies. These would be useful in establishing the results in sections to follow. 

\subsection{Linear Index Code, Capacity Region and \\ Symmetric Capacity}
\label{back1}
\begin{definition}
	[Scalar Linear Index Code \cite{ICIP}] When $\mathcal{S}$ is a finite field, an $(\mathcal{S}, l , \mathcal{R})$ index code is scalar linear if, for the source with $n$ messages, $\mathcal{M} =\{x_1, x_2, \ldots, x_n\}$, the transmitted symbol sequence is given by, $$\mathbb{S}^l=\sum\limits_{j=1}^n V_j x_j.$$
	The $l \times 1$ vector $V_j$ is referred to as the precoding vector (or beamforming vector) for the message $x_j$.
		\label{def1}
\end{definition}

%
%
%

In general, the message $x_i$ is a random variable uniformly distributed over the set $\{1,2,\ldots,\left\vert\mathcal{S}\right\vert^{lR_i}\}$, where $R_i$ is the rate for message $x_i$, $\forall i \in [n]$. If the message $x_i$ consists of $P_i$ symbols, i.e., $x_i \in \mathcal{S}^{P_i}$, the code is said to  achieve a rate  $R_i= \frac{P_i}{l}$ for receiver $D_i$. The rate $R_i$ of a receiver $D_i$  is the normalized amount of information transmitted to it.

Identifying the set of all achievable rate tuples, i.e., the capacity region,  still remains an open problem. However, for several special classes of index coding, for example the symmetric cases when all the rates are equal or when $G$ has certain structure\cite{ICSI}, the capacity region has been fully characterized.

\begin{definition}[Symmetric Rate \cite{ICIP}]
Let $\mathcal{R}$ be the rate tuple $(R_1,R_2,\ldots,R_n)$.
When all  messages have the same rate,  say $\mathbf{R}$, then $\mathbf{R}$ is called the  \textit{Symmetric Rate}. 
\end{definition}

An index code can be \textit{scalar} ($P_i=1; \text{for all}\;  i\in [n]$), \textit{vector} ($P_i>1$ for at least one $i\in [n]$), \textit{linear} (when the encoding function is linear) or \textit{non-linear} (when the encoding function is non-linear). 

In this paper, we deal only with \textit{scalar linear index coding}, which means that each message consists of only one symbol, i.e.,  $P_i=1; \forall i \in [n]$. This implies $\mathcal{R}=(\frac{1}{l}, \frac{1}{l}, \ldots, \frac{1}{l})$ and the encoding function is linear, with $\mathcal{S}$ being a finite field. Clearly, the best rate (symmetric capacity) is achieved when $l=minrank$.
\subsection{Critical Graphs and Index Coding}

\begin{definition}
[Critical Graphs\cite{critical}]

Given a fixed set of rates, let $\mathbb{G}$ denote the set  of all  graphs that support the rates. A graph is said to be critical
(or edge critical) if 
\begin{enumerate}
\item it belongs to $\mathbb{G}$, and 
\item deletion of any
edge from the graph makes it to fall outside $\mathbb{G}$.
\end{enumerate}
 
\label{citicalgraph}
	
\end{definition}

\begin{definition}
	[Strongly Connected Graph\cite{diestel}]A directed graph is strongly connected if there is a path between all pairs of vertices \cite{diestel}.
\label{scg}
\end{definition}

\begin{definition}
	 [Union of Two Disjoint Graphs\cite{ITW}] The union of
	$G = (V ,E)$ and $G' = (V',E')$ is defined as $G \cup G' = (V \cup	V', E \cup E')$.
\end{definition}
\begin{definition}
[Union of strongly connected subgraphs\cite{ITW}] Graph $G$ is the union of
	strongly connected subgraphs if there exists a set of disjoint
	graphs ${G_1,G_2,\cdots ,G_k}$ such that 
	\begin{enumerate} 
	\item $G_i$ is strongly connected,	and 
	\item $G = \underset{i}{\bigcup}\; G_i$.
	\end{enumerate}
\end{definition}

\begin{definition}
	[Hamiltonian Cycle\cite{diestel}] A Hamiltonian path (or traceable path) is a path in a directed/undirected graph that visits each vertex exactly once. A Hamiltonian cycle is a Hamiltonian path that is a cycle.
	 
\end{definition}

\begin{definition}
	[Unicycle\cite{ITW}] A graph $G(V,E)$ is referred to as a unicycle if the set of edges $E$ of the graph is a Hamiltonian cycle of $G$.
		\label{unicycledef}
\end{definition}

\begin{definition}
	[Vertex Induced Subgraph\cite{diestel}] A vertex-induced subgraph  is a subset of the vertices of a graph $G$ together with any edges whose endpoints are both in this subset. 
\end{definition}

\section{Results}
\label{fundaresults}

This section discusses the problem of finding the minrank of a unicast-uniprior problem in detail. The discussions presented in this section lead to a polynomial time algorithm (Section \ref{algo}) for computing minrank.

Consider a unicast-uniprior problem with $N$ receivers. We convert this problem to its single unicast equivalent by splitting the receivers{\cite{KS}}.

 Let $\mathbb{A}_{n\times n}$ be a fitting matrix of this single unicast problem with $n$ receivers and $n$ messages, in its general form (with `$x$' denoting side-information bits). Further,  
  \begin{IEEEeqnarray}{c} 
 n=\sum\limits_{i=1}^{N}|\mathcal{W}_i| = \sum\limits_{i=1}^{N}|\mathcal{A}_i|.
 \end{IEEEeqnarray}
 
 This is because:$(i)$ If $|\mathcal{W}_i| > \sum\limits_{i=1}^{N}|\mathcal{A}_i|$ there exist demands which are not present as side-information at any of the receivers. These messages must be sent as independent transmissions and thus can be excluded while constructing an index code, $(ii)$ If $|\mathcal{W}_i| < \sum\limits_{i=1}^{N}|\mathcal{A}_i|$, there exist as side-information, messages that are no longer demanded by any of the receivers and thus can be excluded while constructing an index code.

 Note that the rows of $\mathbb A$ correspond to receiver indices and  columns correspond to  message indices.

\subsection{Critical Fitting Matrix}
\label{criticalfit}

Matrix theory says that if any one element of a $0-1$ matrix is flipped, one of the following happens to its rank i) increases by one, ii) decreases by one, or iii) remains the same. 

Let $A'$ be the matrix obtained by changing any non-diagonal $1$ of the fitting matrix $A$ whose rank equals minrank.   By definition, $A'$ is also a fitting matrix. Because the rank of $A$ is the minimum among all the fitting matrices, that rank of $A'$  cannot be lesser than the rank of $A$. 
Thus,
\begin{IEEEeqnarray}{c}
\text{rank}(A') \in\{ \text{rank}(A), \text{rank}(A)+1\}. \nonumber 
\end{IEEEeqnarray}

A subset of these fitting matrices whose rank equals minrank shall be called as \emph{critical fitting matrices}.
\begin{definition}
[Critical Fitting Matrix] Given a fitting matrix $\mathscr{A}=[a_{ij}]_{n\times n}$ whose rank is equal to minrank. The fitting matrix $\mathscr{A}$ is called a critical fitting matrix if and only if  changing any  $a_{ij}=1; i\neq j$ to $0$ results in a matrix $\mathscr{A}'$ whose rank is strictly greater than the rank of $\mathscr{A}$. 	
\end{definition}

\subsection{Side-information supergraph}
\label{sidesuper}
In order to fully characterize the unicast-uniprior problem a new graph structure called the \emph{side-information supergraph} is introduced.
\begin{definition}[Side-information supergraph]
	\label{sisuper}

	The side-information supergraph $\mathcal{G(V,E)}$ of a unicast-uniprior problem is defined as follows:
	\begin{enumerate}
		\item A vertex $v_i\in \mathcal{V}$ corresponds to the destination demanding the message $x_i$ for all $i \in [ n ]$. 
		\item The vertices derived by splitting a  demand set $\mathcal{W}_i$ of the original unicast-uniprior problem are grouped into one supernode $\mathscr{S}_i$ for all $i\in [ N ]$.
		\item There exists a directed edge from the supernode $\mathscr{S}_j$, $j\in [ N ]$ to  the vertex $v_i$, $i \in [n]$   iff the receiver $D_j\triangleq (\mathcal{W}_i, \mathcal{A}_i)$ knows the message $x_i$ as side-information.  
		 
	\end{enumerate}
\end{definition}

\subsubsection{Properties of the side-information supergraph}
\label{properties}
We note the following characteristics of the side-information supergraph $\mathcal{G}(V,E)$ of a unicast-uniprior problem:
\begin{itemize}
	\item Every vertex $v_i \in V$, $i \in [n]$ has  at least one incoming edge from a supernode. This is because in the unicast - uniprior problem, we have total number of messages, $n = \sum\limits_{i=1}^{N}|\mathcal{W}_i|= \sum\limits_{i=1}^{N}|\mathcal{A}_i|$. So each message $x_i$ demanded by a receiver $D_k; k \in [ N]$ is known apriori to some receiver $D_j ; j \neq k$. Hence, there is one incoming edge to any  $v_i$ from some supernode $S_j$.
	\item Every $v_i \in V$, $i \in [ n ]$ has  only one incoming edge from a supernode. This is because of the uniprior nature of the problem. More than one incoming edge to a vertex $v_k\; (k \in [n])$, implies more than one receiver $D_i; i \in [N]$ knows the message $x_k$ which is not possible as $\mathcal{A}_i \cap \mathcal{A}_j =\phi, \forall i \neq j ; i,j \in [ N]$.
\end{itemize}

\begin {example}
\label{exmp1}
A sample unicast-uniprior problem is described in Table \ref{tableex1}.
	\begin{table}[!htbp]
	
	\begin{center}
	{
			\begin{tabular}{|c|c|c|c|}
				\hline	
				
				$\mathcal{W}_i$&  $x_1, x_2 $ & $x_3, x_4$ & $x_5$\\
				\hline 
				$\mathcal{A}_i$ & $x_4$& $x_5, x_1$ & $ x_2, x_3 $ \\
				\hline
			\end{tabular}
		}
	\end{center}
	\small\caption{Sample unicast-uniprior problem}
	\label{tableex1}
\end{table} 

The single-unicast equivalent of this problem obtained by splitting the $N=3$ receivers into $n=5$ is shown in Table \ref{ex1tab2}.

	\begin{table}[!htbp]
	
	\begin{center}
	{
			\begin{tabular}{|c|c|c|c|c|c|}
				\hline	
				
				$\mathcal{W}_i$&  $x_1$ & $x_2 $ & $x_3$& $x_4$ & $x_5$\\
				\hline 
				$\mathcal{A}_i$ & $x_4$ & $x_4$ & $x_5, x_1$& $x_5, x_1$ & $ x_2, x_3 $ \\
				\hline
			\end{tabular}
		}
	\end{center}
	\small\caption{Single-unicast equivalent of the unicast- uniprior problem in Table \ref{tableex1}}
	\label{ex1tab2}
\end{table} 
The side-information graph $G$  for this problem  is shown in Fig. \ref{figex2}.

\begin{figure}[htbp]
	\centering
	\includegraphics[scale=0.65]{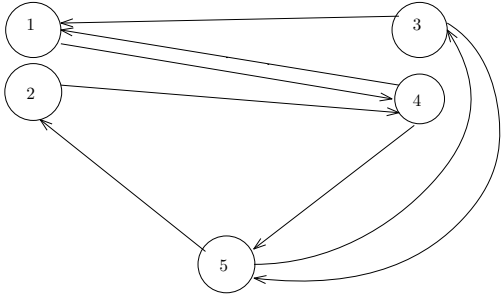}
	\caption{Side-information graph $G$,   problem in Table \ref{ex1tab2} }
	\label{figex2}
\end{figure}

The matrix $A$ in Fig. \ref{exmp1fitfig} is one among the $z=2^{8}=256$ fitting matrices of this problem.    

\begin{figure*}
		{
			\begin{IEEEeqnarray}{ccccccc}
			\mathbb{A}=\left( \begin{array}{ccccc}
			                    1&0&0&x&0\\
                             0&1&0&x&0\\     
                             x&0&1&0&x\\     
                             x&0&0&1&x\\     
                             0&x&x&0&1  
				          \end{array} \right),
         &&\quad & A=\left( \begin{array}{ccccc}
					       1&0&0&1&0\\
                      0&1&0&1&0\\     
                      0&0&1&0&0\\     
                      0&0&0&1&1\\     
                      0&0&1&0&1  
				          \end{array} \right),
			&&\quad &A_c=\left( \begin{array}{ccccc}
			                 1&0&0&1&0\\
					           0&1&0&0&0\\
				              0&0&1&0&1\\
				              1&0&0&1&0\\
				              0&0&1&0&1
				\end{array} \right).
			\nonumber 
			\end{IEEEeqnarray}
		
		\caption{For the problem in Example \ref{exmp1}, $\mathbb{A}$ is the fitting matrix in its general form; $A$ is a fitting matrix; $A_c$ is a critical fitting matrix.}
		\label{exmp1fitfig}
		}
	\end{figure*}

The side-information supergraph $\mathcal{G}$ for this problem is shown in Fig. \ref{figex1}.
	\begin{figure}[htbp]
	\centering
	\includegraphics[scale=0.5]{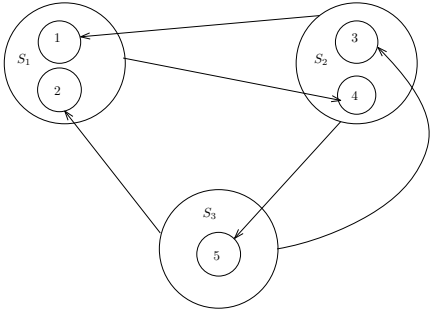}
	\caption{Side-information supergraph $\mathcal{G}$,  problem in Table \ref{tableex1} }
	\label{figex1}
\end{figure}

Notice that the side-information supergraph $\mathcal{G}$ translates to the side-information graph $G$ when each outgoing edge from any supernode $S_i$ is replaced by $|\mathcal{W}_i|$ outgoing edges from the nodes within $S_i$. 

\end{example}

\subsection{On fitting matrices of a unicast-uniprior  problem}
\label{sub1}

Let the subset of rows of $\mathbb A$ that correspond to receivers obtained by splitting the $i^{th}$ receiver of the original unicast-uniprior problem be denoted by $R_{W_i}$. Thus, the $n$ rows of $\mathbb A$ are partitioned into $N$ sets of rows $\{R_{W_1}, R_{W_2}, \ldots, R_{W_N}\}$. 

Henceforth, an arbitrary fitting matrix shall be denoted by $A$, while a fitting matrix in its general form by $\mathbb{A}$.

\begin{proposition}
	\label{prop1}
	Any fitting matrix of a single unicast problem that was derived from a unicast-uniprior problem cannot have more than two identical rows. 
\end{proposition}
\begin{proof}
The proof is given in Appendix \ref{prop1proof}.
\end{proof}
\begin{corollary}
	\label{cor1}
	The side-information graph $G$ of a single unicast problem that is derived from a unicast-uniprior problem cannot have cliques with three or more vertices.
\end{corollary}
\begin{proof}
The proof is given in Appendix \ref{proofcor1}.
\end{proof}
 \begin{definition}
 \label{minimaldefn}
A set of linearly dependent rows is said to be \emph{minimally dependent} if every proper subset of it is linearly independent. 
\end{definition}
\begin{proposition}
	\label{prop2}
	Delete every pair of identical rows from $A_{n\times n}$ to obtain  $A'_{m\times n}$. Suppose that the rank of $A'$ is $r'$. Then there are at least  $m-r'$ minimally dependent sets of rows in $A'$ and at least  $n-r$ minimally dependent sets of rows in $A$.	
\end{proposition}
\begin{proof}
The proof is given in Appendix \ref{proofprop2}.
\end{proof}

Note that the receiver demanding the message $x_j$ corresponds to the $j^{th}$ row $\underbar{r}_j$ of $\mathbb{A}$, where $j\in[n]$. 

\begin{proposition}
	\label{prop3}
	The message demanded by the $j^{th}$ receiver, $x_j$ is present as side-information in one and only one $R_{W_i}$ where $i \neq j$. 
\end{proposition}

\begin{proof}
	The proof directly follows from the uniprior nature of the problem.
\end{proof}

\begin{lemma}
	\label{lemma1}
	
	Given a set  $L$ of $l$ rows of the  matrix $\mathbb A$. Let it  be comprised of $c_i$ rows from  $R_{W_i}$,  $i \in [ N ]$.  Suppose that in some  fitting matrix these $l$ rows are minimally dependent. Then, there always exists a subset $\mathbb{S}$ of these rows that are minimally dependent with at most one row from any $R_{W_i}$, for a different choice of \enquote{$x$}s in these rows. 
\end{lemma}
\begin{proof}
The proof is given in Appendix \ref{prooflemma1}.
\end{proof}

 \begin{lemma}
 	\label{lemma2}
 	The vertex induced subgraph $Q$ of  $G$, whose vertices correspond to the rows of the set  $\mathbb{S}$ in Lemma \ref{lemma1} forms a unicycle.
 \end{lemma}
\begin{proof}
The proof is given in Appendix \ref{prooflemma2}.
\end{proof}

\subsection{Identifying the critical side-information}
\label{sub2}
 
 For an  index coding problem, the rank of  a fitting matrix that has $x=1$ for all the critical side-information bits  and $x=0$ for all the non-critical side-information bits, equals its minrank. When one is interested in computing the minrank, knowledge of critical graphs naturally helps in narrowing down the number ($z$) of fitting matrices whose rank must be computed.
 
The following  theorem  in \cite{ITW} helps in  identifying the critcial edges in a side-information graph.  
\begin{theorem}
	\label{criticaledge}
	An edge $e$ in the side-information graph $G$ is critical if it belongs to a vertex induced subgraph of $G$ which is a unicycle \cite{ITW}.
\end{theorem}

For the unicast-uniprior problem, identifying critical side-information is as good as identifying the unicycles in the side-information graph $G$. This is discussed in the following section.
\subsection{Translating critical graphs into fitting matrices for a unicast-uniprior problem}
\label{sub3}


Next, we combine the results in Sections \ref{sub1} and \ref{sub2} into the context of a unicast-uniprior problem and subsequently develop the algorithm for finding minrank in Section \ref{algo}. 

Using Proposition \ref{prop1}, any fitting matrix of a unicast-uniprior problem cannot have more than two identical rows. Using Proposition \ref{prop3}, each of these identical rows belong to $R_{W_i}$ and $R_{W_j}$ with $i \neq j$;  $i,j \in[N]$.

Using Proposition \ref{prop2}, it is clear that a fitting matrix has at least $m-r$ minimally dependent rows. For any critical fitting matrix  $r=minrank$ and hence it is guaranteed to have at least $m-r$ minimally dependent sets of rows (which is the guaranteed maximum according to Proposition \ref{prop2}). 

The following theorem \cite{critical} regarding critical graphs becomes important in identifying the critical fitting matrices for a unicast-uniprior problem.

\begin{theorem}
	\label{USCS}
	Every critical graph for linear index coding is a union of strongly connected subgraphs. In particular, removing the edges not lying on a
	directed cycle does not change the capacity region in
	these cases \cite{critical}.
\end{theorem}


\section{Algorithm}
\label{algo}
This section discusses an efficient algorithm to compute the minrank of a unicast-uniprior problem in polynomial time. 


\subsection{Algorithm to compute minrank}
\begin{enumerate}
	\item Let $W_{max}\triangleq \underset{i=[N]}{max} \mid\mathcal{W}_i\mid$.
	\item We know that  the $(i,i)^{th}$ element  of a fitting matrix corresponds to the $i^{th}$ wanted message $x_i$. We form an $N \times W_{max} $ table $T$ such that the $k^{th}$ row of the table consists of messages in $\mathcal{W}_k; k \in [N]$. Clearly, there are 
	\begin{IEEEeqnarray}{c}
	\beta=\prod_{i=1}^{N} {{W_{max}} \choose {|\mathcal{W}_i|}} |\mathcal{W}_i|!
	\end{IEEEeqnarray}  ways of filling this table. Let each instance of $T$ be denoted by $T_i; i\in [\beta]$.
	\item Consider each column of the table as a single unicast uniprior problem. The wanted message is the table entry and the side-information is same as the side-information derived from the unicast-uniprior problem. The minrank of a single unicast uniprior index coding problem can be computed in linear time as discussed in \cite{KNR}.
	
	\item Let the minranks along each column in $T_i$ be denoted by $rk_{i1}, rk_{i2}, \ldots, rk_{iW_{max}}$. 
	We define the sum of minranks of the $W_{max}$ single unicast uniprior problems as follows: For $i \in [\beta]$,
	
	\begin{IEEEeqnarray}{c}
	S_i \triangleq \sum\limits_{j=1}^{W_{max}} rk_{ij}
	\end{IEEEeqnarray}
	
	\item The minrank of the unicast-uniprior problem is given by: 
	\begin{IEEEeqnarray}{c}
	minrank = {min}{\;S_i}.
	\label{Sieqn}
	\end{IEEEeqnarray}
\end{enumerate}

\subsection{Proof of correctness}
\label{correctness}

\begin{lemma}
\label{ch4lemma3}
 Let $Q$ be a unicycle in the side-information graph $G$. Let the vertices of $Q$ be part of a bigger directed cycle $C$. Let $S_Q$ and $S_C$ denote the sets of rows in a fitting matrix $A$ that correspond to the vertices of $Q$ and $C$, respectively. If the rows in $S_Q$ are minimally dependent in $A$, then the rows in $S_C$ are not minimally dependent and vice versa. 
\end{lemma}
\begin{proof}
The proof directly follows from the definition of a minimally dependent set (Definition \ref{minimaldefn}).
\end{proof}

\begin{lemma}
Let $G_C$ be  subgraph of $G$ induced by the vertices of  a directed cycle $C$. The subgraph $G_C$ has one or more disjoint vertex induced subgraphs $G_{Q_1}, G_{Q_2}, \ldots, G_{Q_m}$ that are unicycles.
\end{lemma}
\begin{proof}
Consider Example \ref{exmp1}. The vertices $C=(v_1, v_4, v_5, v_3)$ form a directed cycle in $G$. However, there exists two vertex induced subgraphs of vertices in $C$ which are unicycles, viz., $Q_1=(v_1, v_4)$ and $Q_2=(v_3, v_5)$.
\end{proof}

Theorem \ref{USCS} states that a critical graph is always a union of strongly connected subgraphs and  removing the edges not lying on a directed cycle in a critical graph do not change the capacity region.

Theorem \ref{criticaledge} says that an edge is critical if it belongs to a unicycle. Thus, by keeping just the edges that form  unicycles in $G$, a critical graph could be constructed, followed by a critical fitting matrix.

In any fitting matrix of rank $r$ there are necessarily $n-r$ sets of minimally dependent rows (Proposition \ref{prop2}). Note that $n-r$ is maximum when $r=minrank$. 
Lemma \ref{lemma1} and Lemma \ref{lemma2} together imply that given a minimally dependent set of rows, there always exist a unicycle $Q$ in $G$ with at most one vertex (row) from $R_{W_i}; i\in [N]$. Further, if that minimally dependent set corresponds to a directed cycle $C$ in $G$, Lemma \ref{ch4lemma3} implies that in any fitting matrix where $S_Q$ is minimally dependent, $S_C$ is not. This justifies Step $2$ of the algorithm where each column consists of at most one demand from any $\mathcal{W}_i; i\in [N]$.


One unicycle is equivalent to one minimally dependent set of rows and hence reduces the rank by one. Step $3$ of the algorithm computes minrank of single unicast uniprior problems, using the method in \cite{KNR}, which prunes the corresponding side-information graph to obtain the maximum number of disjoint unicycles.

This justifies Step $4$ of the algorithm where summing the minranks (along each column of the table) is analogous to taking the count of unicycles after one round of pruning $G$. Note that, minimizing the ranks of every column for one $T_i$ constitute one round of pruning $G$.

Finally, minimizing the sum of minranks over all $T_i; i\in [\beta]$ gives the minrank  of the unicast-uniprior problem.  Minrank is obtained when we have pruned $G$ to get the maximum possible number of disjoint unicycles.

\subsection{ Identifying critical side-information from the algorithm}

Suppose that $S_i=minrank$. The critical side-information for the unicast uniprior problem can be characterized  using table $T_i$. Let $\mathscr{C}_1, \mathscr{C}_2, \ldots, \mathscr{C}_{W_{max}}$ denote the columns of $T_i$. 

\begin{enumerate}
\item Along every column $\mathscr{C}_i; i\in [W_{max}]$, there is a single unicast uniprior problem whose rank can be computed in polynomial time. Let the fitting matrix of minimum rank  for each of these $W_{max}$ problems be denoted by $\mathscr{A}_1, \mathscr{A}_2, \ldots, \mathscr{A}_{W_{max}}$.  
 
\item  Every minimally dependent set of rows in these matrices form unicycles in the side-information graph $G$. Construct $G'$ from $G$ by deleting all the edges other than those that constitute those unicycles. 
$G'$ represents a critical graph for the concerned unicast-uniprior problem and the side-information bits corresponding to the edges in $G'$ are critical.

\end{enumerate}

\section{Algorithm In Action}
\label{examples}
\begin{example}
	\label{ex2}
Consider the following unicast-uniprior index coding problem:
There are $N=5$ receivers and $n=10$ messages. The receiver demands and side-information are given in Table \ref{ex1tab}.

	\begin{table}[!htbp]
		
		\begin{center}
			{
				\begin{tabular}{|c|c|c|c|c|c|}
					\hline	
					
					$\mathcal{W}_i$&  $x_1, x_2 $ & $x_3, x_5$ & $x_4, x_6$& $x_9$ & $x_7, x_8, x_{10}$\\
					\hline 
					$\mathcal{A}_i$ & $x_3$& $x_1,x_4$ & $ x_5, x_9, x_8$& $x_{10}, x_7$& $x_2, x_6$ \\
					\hline
				\end{tabular}
			}
		\end{center}
		\small\caption{ Unicast-uniprior problem in Example \ref{ex2}}
		\label{ex1tab}
	\end{table} 

We illustrate the step-by-step execution of the algorithm and compute the minrank for this problem.
\begin{enumerate}  
\item { $W_{max}= \underset{i=[ 5 ]}{max}\; |\mathcal{W}_i|=3 $ .}
\item{ We split the given unicast-uniprior problem into $W_{max}=3$ single-unicast uniprior problems. There are $\beta$ ways of doing this, where $${\beta =\prod_{i=1}^{5}} {3 \choose { |\mathcal{W}_i|}} |\mathcal{W}_i|!=1944.$$ All of these $1944$ ways of splitting can be captured by filling a $5 \times 3$ table. First instance out of the $1944$ possible splitting is given below:
	\begin{table}[!htbp]
	
	\begin{center}
		\scriptsize{
			\begin{tabular}{|c|c|c|}
				
			  \hline	
				$I$& $II$ & $III$\\
					\hline
				 $x_1 $ & $x_2$ & \\
				\hline 
				$x_3$ & $x_5$& \\
				\hline
				$x_4$ & $x_6$ & \\
				\hline
				$x_9$ & & \\
				\hline 
				$x_7$ & $x_{10}$ & $x_8$\\
				\hline 
			\end{tabular}
		}
	\end{center}

\end{table}

 }
\item{ The three single-unicast uniprior problems that are obtained from the above table are as follows:
	
		\begin{table}[!htbp]
		
		\begin{center}
			\scriptsize{
				\begin{tabular}{|c|c|c|c|c|c|}
					
					\hline	
						\multicolumn{2} {|c|} {Problem $I$}& \multicolumn{2} {c|} {Problem $II$ }& \multicolumn{2} {c|} {Problem $III$}\\
					\hline
					$\mathcal{W}_1$ & $\mathcal{A}_1$ & $ \mathcal{W}_2$ & $ \mathcal{A}_2$ & $ \mathcal{W}_3$ & $ \mathcal{A}_3$\\
					\hline
					$x_1 $ & $x_3$& $x_2$ &$x_3$ &&\\
					\hline 
					$x_3$ & $x_1, x_4$&  $x_5$&$x_1, x_4$ & &\\
					\hline
					$x_4$ &$x_5, x_9, x_8$& $x_6$ & $x_5, x_9, x_8$&& \\
					\hline
					$x_9$ &$x_{10}, x_7$ & &&&\\
					\hline 
					$x_7$ & $x_2, x_6$& $x_{10}$ & $x_2, x_6$ & $x_8$&$x_2, x_6$\\
					\hline 
				\end{tabular}
			}
		\end{center}
		
	\end{table}

}

\item{We evaluate the minranks of the three problems using the method given in \cite{KNR}. $rk_{11}=4, rk_{12}=4, rk_{13}=1$. $$S_1 = rk_{11}+ rk_{12}+ rk_{13}=9$$ }

\item{ Similarly, we find $S_i$ values for all $z=1944$ ways of splitting the original problem. After that we get the minrank as follows:
$$minrank \triangleq \underset{i=[ 1944] }{min}{S_i}=7.$$
}
\end{enumerate}

\subsection{Analysis for reduction in complexity}
We know that calculating minrank by brute force requires evaluating the $\mathbb{F}_2$ rank of $2^{\sum\limits_{i=1}^{n} |\kappa_i|}$ matrices and finding the minimum among them. If Gaussian elimination is used to compute rank, this method has a computational complexity of  $2^{\sum\limits_{i=1}^{n}|\kappa_i|} O(n^2)$. 
In this example instead of evaluating the ranks of $2^{10}=1024$ matrices each of order $n=10$, we evaluate the ranks of just $1944$ matrices, each of order $N=5$ to arrive at minrank. Further the brute force technique employs gaussian elimination whose computational complexity is $O(n^2)$ per matrix. The proposed technique computes ranks of each of the $1944$ matrices in linear polynomial number of computations the complexity of which is $O(N)$.

Thus we see huge computational savings with the proposed method of finding minrank. As $n$ increases, the new method shall offer a greater computational advantage over the brute force technique.

\end{example}

 \begin{example}
 \label{ex3}
 Consider the problem described in Table \ref{ex2tab}.
 
 	\begin{table}[!htbp]
 	
 	\begin{center}
 		\scriptsize{
 			\begin{tabular}{|c|c|c|c|c|c|}
 				\hline	
 				
 				$\mathcal{W}_i$&  $x_1, x_6, x_{10} $ & $x_2, x_4, x_7, x_9$ & $x_3, x_5, x_8$\\
 				\hline 
 				$\mathcal{A}_i$ & $x_2,x_4,x_9$& $x_3, x_5, x_8, x_{10}$ & $ x_1, x_6, x_7$ \\
 				\hline
 			\end{tabular}
 		}
 	\end{center}
 	\small\caption{ Unicast-uniprior problem in Example \ref{ex3}.}
 	\label{ex2tab}
 \end{table} 
 The minrank for this problem is $7$.
 If we compute minrank by brute force, we need to compute the ranks of  $2^{(9+16+9)}=2^{34} \approx 1.717 \times 10^{10}$ matrices, each of order $10 \times 10$. Notice that the rank over $\mathbb{F}_2$ is found by gaussian elimination. When we employ our algorithm we compute the minranks of $\beta=13824$ single unicast uniprior problems. The ranks are found in number of computations that is linear in the number of messages. Simulation result  that computed minrank using our algorithm was found to agree with that obtained from brute force computations.
 
\end{example}

\begin{example}
	\label{ex4}

		Consider the problem described in Table \ref{ex3tab}.
		
		\begin{table}[!htbp]
			
			\begin{center}
				\scriptsize{
					\begin{tabular}{|c|c|c|c|c|c|}
						\hline	
						
						$\mathcal{W}_i$&  $x_1, x_2, x_3, x_{4} $ & $x_5, x_6, x_7$ & $x_8, x_9, x_{10}, x_{11}, x_{12}$\\
						\hline 
						$\mathcal{A}_i$ & $x_5, x_9, x_{10}, x_{11}, x_{12}$ & $x_1, x_2, x_8$ & $ x_3, x_4, x_6, x_7$ \\
						\hline
					\end{tabular}
				}
			\end{center}
			\small\caption{ Unicast-uniprior problem in Example \ref{ex4}.}
			\label{ex3tab}
		\end{table} 
		The minrank for this problem is $8$.
	Brute force computation of minrank needs computing the ranks of $2^{(20+9+20)}=2^{49} \approx 5.629 \times 10^{14}$ matrices (each of order $12 \times 12$) by Gaussian elimination. When our algorithm is employed, we compute the minranks of $\beta=864000$ single unicast uniprior problems. Simulation result that computed minrank using our algorithm was found to agree with that obtained from brute force computations.
		
\end{example}

\section{Conclusion}
\label{concl}
This work solves the unicast-uniprior index coding problem. Novel ideas like critical fitting matrix and side-information supergraph were employed to prove the results. A discussion on the properties of the fitting matrix  led to identifying the critical  side-information bits. Also,  a novel method to compute the minrank is provided. The proposed technique  greatly simplifies the existing brute force method of computing minrank.
\appendix

\subsection{Proof of Proposition \ref{prop1}}
	\label{prop1proof}
	Consider a  fitting matrix $A$. Suppose that $p>2$ rows of $A$ are identical. We prove by contradiction that this cannot be true.
	
	Two or more rows from a given $R_{W_i}$ have identical side-information and thus have \enquote{$x$}s in the same columns. However each row has a $1$  as its diagonal entry corresponding to the wanted message. 
	
	Suppose $p=3$ rows of $A$ are identical. For this to be true, each row must be taken from a different $R_{W_i}$. Let these be $R_{W_{\alpha_1}}, R_{W_{\alpha_2}}, R_{W_{\alpha_3}}$. As the problem is uniprior, $\mathcal{A}_{\alpha_m}\cap \mathcal{A}_{\alpha_n}=\phi$ whenever $\alpha_m \neq \alpha_n$. The columns having \enquote{$x$}s for  one row will always be different from the columns that have \enquote{$x$}s in another.  Due to the unicast-uniprior nature of the problem, the message demanded by one receiver cannot be present as side-information at more than one $R_{W_i}$s. So if we need $3$ identical rows, the \enquote{$x$}s in these rows must be $0$s. But then, we already have \enquote{$1$}s in the diagonal entries of all the $3$ rows corresponding to their wanted messages. This contradicts the assumption that $p=3$ rows of $A$ are identical.	
	This implies that for any fitting matrix $A$, $p>2$ rows can never be identical. 
	
\subsection{Proof of Corollary \ref{cor1}}

\label{proofcor1}

	Assume that the side-information graph has a clique with $p>2$ vertices. Consider the $p$ rows of $\mathcal{A}$ that correspond to these vertices. Let them be $\underbar r_{\alpha_1}, \underbar r_{\alpha_2}, \ldots, \underbar r_{\alpha_p}$. These rows will have $1$s in their diagonal entries corresponding to the wanted messages. In particular, we choose the \enquote{$x$}s in these rows such that $x=1$ in the $(\alpha_i, j)^{th}$ entry only if  the directed edge $e_{j,\alpha_i}$ is part of the clique; $x=0$ otherwise.  We know that in a clique one vertex has incoming edges from the remaining $p-1$ vertices. The entries corresponding to all other incoming edges are chosen to be $0$s. Thus we get a set of $p$ identical rows (they have $1$s in columns indexed by the $p$ vertices of the clique and $0$s elsewhere). This implies the existence of a fitting matrix with more than two identical rows, contradicting Proposition \ref{prop1}. Hence, our assumption that the side-information graph has a clique with $p>2$ vertices is wrong. Hence proved.
	
\subsection{Proof of Proposition \ref{prop2}}
\label{proofprop2}
	Given that $A'_{m \times n}$ is the matrix obtained after removing all pairs of identical rows from $A_{n \times n}$. When $A'$ is of rank $r'$, there are $r'$ linearly independent rows and each of the remaining $m-r'$ rows is a linear combination of these $r'$ rows.  Let the rows of the matrix $A'$ be re-arranged to get the matrix $B'$ such that the first $r'$ rows of $B_{m \times n}$ are linearly independent. Let the rows of $B_{m \times n}$ be represented by $\underbar r_1, \underbar r_2,\ldots, \underbar r_m$. Each of the rows $\underbar r_{r'+1}, \underbar r_{r'+2}, \ldots, \underbar r_{m} $ is some linear combination of the first $r'$ rows, $\underbar r_1, \underbar r_2, \ldots, \underbar r_{r'}$. Thus we have,
	$$B=\left[\begin{array}{lllllll}
	\underbar r_1&	    \underbar r_2&
	\cdots &
	\underbar r_r&
	\sum\limits_{i=1}^{r} c_{ji} \underbar r_i&
	\cdots&
	\sum\limits_{i=1}^{r} c_{ji} \underbar r_i
	\end{array}\right]^T$$
	where $c_{ji} \in \{0,1\}; i \in [ r' ]$.
	
	The row $\underbar r_j ; j\in \{r'+1, \ldots, m\}$ along with the set of linearly independent rows that add to give $\underbar r_j$ (the set of rows $\underbar r_j \cup \{\underbar r_i ;\; i \in [ r' ]$ where $c_{ji}=1$\}), form a minimally dependent set of rows. Clearly, we have $m-r'$ minimally dependent sets of rows in $A'$.
	For the matrix $A$, each pair of identical rows represents one minnimally dependent set. So, $A$ has at least $n-r$ sets of minimally dependent rows.


\subsection{Proof of Lemma \ref{lemma1}}
\label{prooflemma1}
	Consider a set of minimally dependent rows $R=\{\underbar{r}_1, \underbar{r}_2, \ldots, \underbar{r}_l\}$ taken from a fitting matrix $A$. Let them come from $\Omega \leq N$ different demand sets of the unicast-uniprior problem. Let these rows form an $l \times n$ matrix $L$. Each of these rows has $1$ in the column corresponding to the message demanded (diagonal entries of the fitting matrix). Also, each column shall have an even number of $1$s, because if not, the rows cease to be minimally dependent.

In  any row $ \underbar{r}_i$, in the $i^{th}$ column, apart from the $1$ in the $(i,i)^{th}$ position (denoting the message demanded), there is always an odd number of $1$s. This is because, if not, the $l$ rows cease to be minimally dependent. Also, because of the uniprior nature of the problem these odd number of $1$s in any given column occur in one and only one $R_{W_i}$ (Proposition \ref{prop3}).

In those columns whose indices do not correspond to the messages demanded in the $l$ rows, there is always an even number of $1$s and as the problem is uniprior, these $1$s occur in the rows derived from the same demand set (some $R_{W_i}$, out of the $\Omega$ different ones that constitute the minimally dependent set of rows). We choose to ignore these columns and focus on the remaining columns. Thus, we form a subgraph $\mathcal{G}'$ of the side-information supergraph $\mathcal{G}$, keeping only those $l$ nodes corresponding to the messages demanded in the $l$ rows.

In $\mathcal{G}'$, locate the vertex that corresponds to the message demanded in the first row $\underbar{r}_1$, let it be $A_1$. WLOG, assume that it is part of the supernode $S_1$. There will be an odd number of $1$s in the same column as the $1$ corresponding to the message demanded in $\underbar{r}_1$. This means there is an incoming edge from a supernode $S_j$, $j \in [ N ], j\neq 1$, to the node $A_1$. Pick one node (say $A_2$) in $S_j$. Let the row coresponding to $A_2$ be denoted by $\underbar{r}_2$.  Since $R$ is minimally dependent, there will be an odd number of $1s$ in the same column as the $1$ corresponding to the message demanded in $\underbar{r}_2$. This means, there is some supernode $S_i$ that has an outgoing edge to node $A_2$. This $S_i$ could be one among the $\Omega-1$ supernodes in $\mathcal{G}'$ other than $S_2$.

If it is $S_1$, we have found the subset of rows as $\underbar{r}_1, \underbar{r}_2$ with $x=1$ in $\underbar{r}_1$ in  the column corresponding to the message demanded in $\underbar{r}_2$ and $x=1$ in $\underbar{r}_2$ in the column corresponding to the message demanded in $\underbar{r}_1$.

If not $S_1$, it could be some other supernode $S_j$, $j \neq 1,2$. Again, we can pick one node say $A_3$ in $S_j$ and find that supernode $S_i$ which has an outgoing edge to $A_3$. But there are only a finite number of supernodes ($S$), and we know that the number of minimally dependent rows (or the number of vertices in $\mathcal{G}$) $l>\Omega$  as we have more than one row from at least one $R_{W_i}$. Hence, as we proceed, there will be an incoming edge from some supernode $S_i$ which contains a node that was previously traversed, say $A_i$, to the node in hand, say $A_f$. 

When this happens, we can start with the supernode $S_i$ of the final node to which there was an incoming edge ($A_f$), and traverse in the reverse direction along the path previously traversed, getting back to  $A_f$, forming a cycle. Now, pick the rows corresponding to the vertices in this cycle. For each row in this cycle, if we put $x=1$ in the column corresponding to the message demanded by the node where an edge ends, we get a set of minimally dependent rows with at most one from any $R_{W_i}$. 

\begin{example}
	\label{exlemma}
	This example simplifies the understanding of Lemma \ref{lemma1} and its proof.
		Consider the  unicast-uniprior problem as given in Table \ref{table2}. Clearly, $N=5, n=15$.

		\begin{table}[!htbp]
		
		\begin{center}
			{
				\begin{tabular}{|c|c|c|c|c|c|}
					\hline	
					
					$\mathcal{W}_i$&  $x_1, x_2 $ & $x_5, x_6,$ & $x_8, x_9,$ & $x_{11}, x_{12}, $ & $x_{15}$\\
					
					&$ x_3, x_4$& $x_7$ & $x_{10}$&$x_{13}, x_{14}$&\\
					\hline 
					$\mathcal{K}_i$ & $x_5, x_8, $& $x_1,x_2, $ & $ x_3,x_6$ & $x_4, x_7$ & $x_9, x_{10}, x_{13}$ \\
					&$x_{11}, x_{15}$&$x_{12}, x_{14}$&&&\\
					\hline
				\end{tabular}
			}
		\end{center}
		\small\caption{Unicast-uniprior problem; Example \ref{exlemma}.}
		\label{table2}
	\end{table}  
		
		A set $L$ of minimally dependent rows is given below  with $c_1=4, c_2=2, c_3=1, c_4=2, c_5=1$.  Thus, $l=10, \Omega=5$. The matrix $\mathbb{L}$ represents the rows $\{\underbar{r}_1, \underbar{r}_2, \underbar{r}_3, \underbar{r}_4, \underbar{r}_5, \underbar{r}_7, \underbar{r}_8, \underbar{r}_{11}, \underbar{r}_{13}, \underbar{r}_{15}\}$ from an arbitrary  fitting matrix.	
	   \[\scriptsize \mathbb{L}=
	   \begin{array}{l}
	   \;\begin{array}{rrlllllrrl}
	    1 & 2 & 3 & 4 & 5 & 6 & 7 & 8 & 9 & 10 \;\;\; 11 \;\;\; 12 \;\;\; 13 \;\;\;14 \;\;\; 15\\
	    \hdashline\\
	   \end{array}\;\\
	   \left[\begin{array}{llllllllllllllllll}
	    1 & 0 & 0 & 0 & 1 & 0 & 0 & 0 & 0 & 0 & 0 & 0 & 0 & 0 & 0 \\
		0 & 1 & 0 & 0 & 1 & 0 & 0 & 1 & 0 & 0 & 0 & 0 & 0 & 0 & 0\\	
		0 & 0 & 1 & 0 & 1 & 0 & 0 & 0 & 0 & 0 & 0 & 0 & 0 & 0 & 1 \\

		0 & 0 & 0 & 1 & 0 & 0 & 0 & 0 & 0 & 0 & 1 & 0 & 0 & 0 & 0 \\
			\hdashline
		1 & 1 & 0 & 0 & 1 & 0 & 0 & 0 & 0 & 0 & 0 & 0 & 0 & 0 & 0 \\
		0 & 0 & 0 & 0 & 0 & 0 & 1 & 0 & 0 & 0 & 0 & 0 & 0 & 0 & 0 \\
			\hdashline
		0 & 0 & 1 & 0 & 0 & 0 & 0 & 1 & 0 & 0 & 0 & 0 & 0 & 0 & 0 \\
		
			\hdashline
		0 & 0 & 0 & 0 & 0 & 0 & 1 & 0 & 0 & 0 & 1 & 0 & 0 & 0 & 0\\
		0 & 0 & 0 & 1 & 0 & 0 & 0 & 0 & 0 & 0 & 0 & 0 & 1 & 0 & 0 \\
			\hdashline
		0 & 0 & 0 & 0 & 0 & 0 & 0 & 0 & 0 & 0 & 0 & 0 & 1 & 0 & 1\\  
		\end{array}\right]\\
		\end{array}\]

\begin{figure}[htbp]
	\centering
	\includegraphics[scale=0.4]{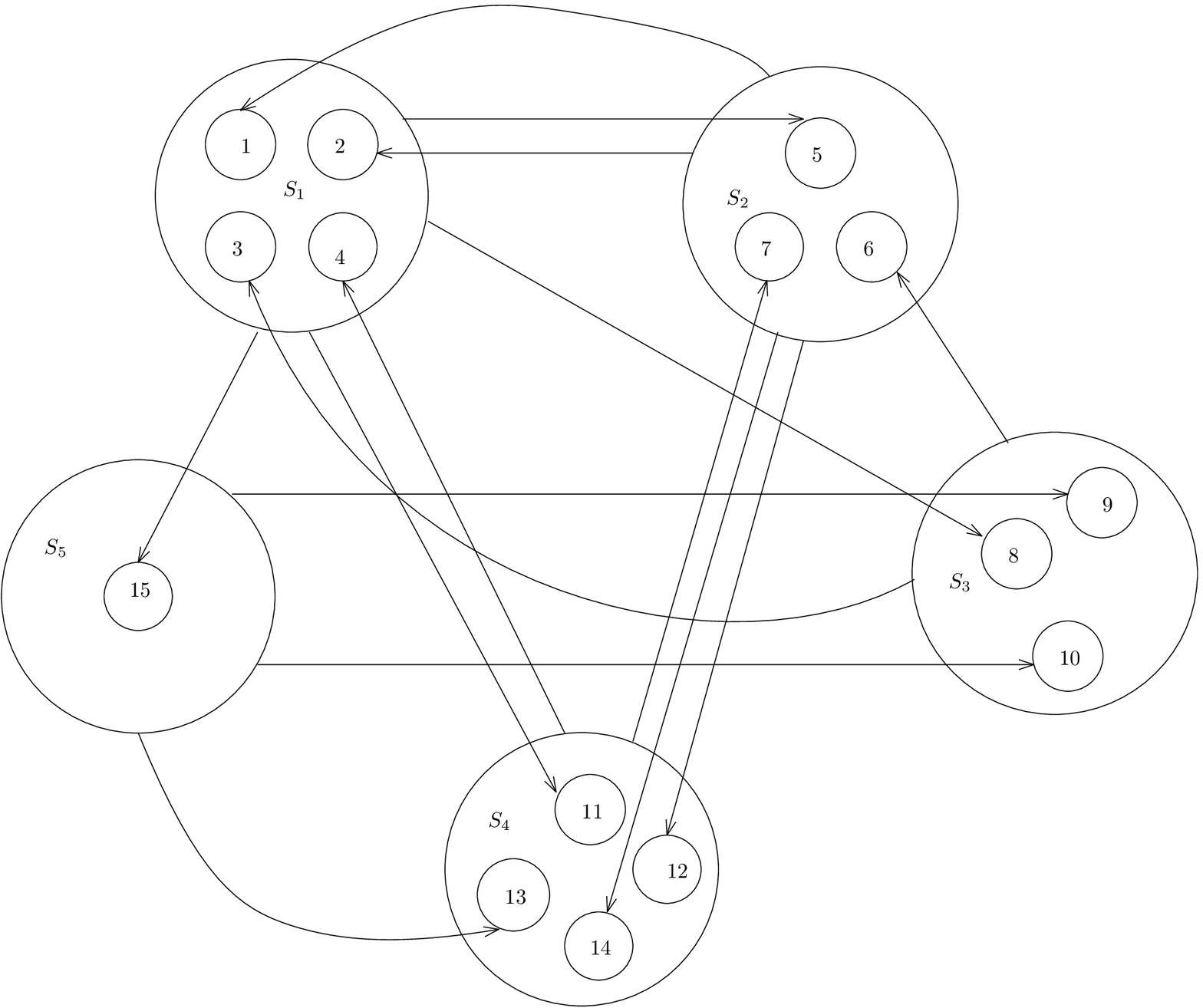}
	\caption{$\mathcal{G}$ for Example \ref{exlemma}.}
	\label{supgr2}
\end{figure}

 $\mathcal{G}'$ obtained from the side-information supergraph	$\mathcal{G}$ (Fig. \ref{supgr2}) by keeping just the $10$ vertices  corresponding to the messages demanded, and their supernodes is shown in Fig. \ref{fig2ex2}.

		\begin{figure}[htbp]
		\centering
		\includegraphics[scale=0.4]{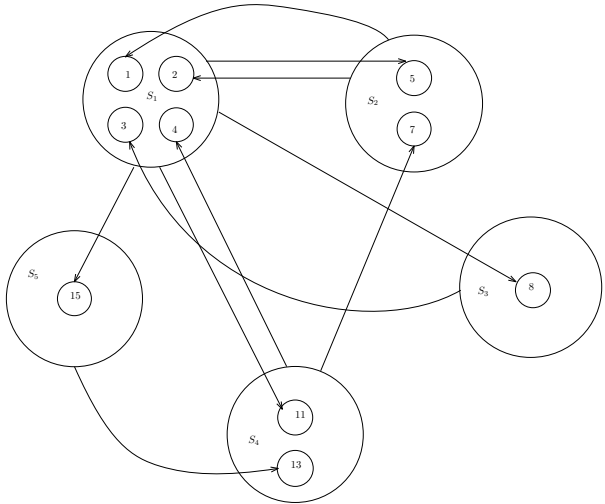}
		\caption{$\mathcal{G}'$ for Example \ref{exlemma}.}
		\label{fig2ex2}
	\end{figure}
%
%
	
	
%
	
	WLOG we start from  $v_1$. There is an incoming edge from the supernode $S_2$ to $v_1$. Now pick any one vertex from $S_2$, say $v_5$. There is an incoming edge from $S_1$ to $v_5$. Hence we get a subset of two minimally dependent rows, with at most one from any $R_{W_i}$, as follows:
	$$\begin{array}{llllllllllllllllll}
	\underbar{r}_1 & 1 & 0 & 0 & 0 & 1 & 0 & 0 & 0 & 0 & 0 & 0 & 0 & 0 & 0 & 0 \\
	\hdashline
	\underbar{r}_5 &	1 & 0 & 0 & 0 & 1 & 0 & 0 & 0 & 0 & 0 & 0 & 0 & 0 & 0 & 0 \\

	\end{array}$$

	Alternately, start from $v_1$. There is an incoming edge to $v_1$ from $S_2$. Pick one vertex from $S_2$, say $v_7$. There is an incoming edge to $v_7$ from $S_4$. Pick one vertex from $S_4$, say $v_{11}$. There is an incoming edge to $v_{11}$ from $S_1$. Note that $S_1$ is the supernode of the vertex we started with, i.e., $v_1$. Hence we get a subset of three minimally dependent rows, with at most one from any $R_{W_i}$, as follows: 
		$$   \begin{array}{lllllllllllllllllll}
	\underbar{r}_1 &1 & 0 & 0 & 0 & 0 & 0 & 0 & 0 & 0 & 0 & 1 & 0 & 0 & 0 & 0 \\
	
	\hdashline

	\underbar{r}_7 & 1 & 0 & 0 & 0 & 0 & 0 & 1 & 0 & 0 & 0 & 0 & 0 & 0 & 0 & 0\\
	\hdashline
	\underbar{r}_{11}& 0 & 0 & 0 & 0 & 0 & 0 & 1 & 0 & 0 & 0 & 1 & 0 & 0 & 0 & 0\\  
	
	\end{array}.$$
	
\end{example}

\subsection{Proof of Lemma \ref{lemma2}}
\label{prooflemma2}
 Using the proof of Lemma \ref{lemma1} in Appendix \ref{prooflemma1}, we know that the vertices corresponding to the set of rows, $\mathbb{S}$ form a cycle. The set $\mathbb{S}$ consists of at most one row from any $R_{W_i}$. Using Proposition \ref{prop3}, in the vertex induced subgraph $Q$ each vertex can have at most one outgoing edge. Hence, $Q$ is a unicycle (see Definition \ref{unicycledef}).
      
 \balance
	\bibliographystyle{IEEEtran}
	
\end{document}